\documentclass[11pt,letterpaper]{article}
\usepackage[lmargin=1.0in,rmargin=1.0in,bottom=1.0in,top=1.0in,twoside=False]{geometry}

\usepackage{fullpage,amssymb,amsmath}
\usepackage{graphicx}
\usepackage{enumerate}
\usepackage{tikz}
\usetikzlibrary{shapes}

\usepackage[T1]{fontenc}

 \usepackage{xcolor}
 \usepackage{mathtools}
\usepackage{microtype}
\usepackage{amsfonts}
\usepackage{comment}
\usepackage[english]{babel}
\usepackage{mathrsfs}
\usepackage[ruled,vlined]{algorithm2e}
\usepackage{blindtext}

\usepackage{trimspaces}
\usepackage{nccfoots}
\usepackage{setspace}
\usepackage{inconsolata}
\usepackage{libertine}
\usepackage[absolute]{textpos}

\usepackage{enumitem}
\usepackage{todonotes}

\usepackage{longtable}

\definecolor{blue}{rgb}{0.1,0.2,0.5}
\definecolor{brown}{rgb}{0.6,0.6,0.2}
\usepackage[ocgcolorlinks, linkcolor={blue}, citecolor={brown}]{hyperref}

\usepackage[amsmath,thmmarks,hyperref]{ntheorem}
\usepackage{cleveref}

\crefformat{page}{#2page~#1#3}%
\Crefformat{page}{#2Page~#1#3}%
\crefformat{equation}{#2(#1)#3}%
\Crefformat{equation}{#2(#1)#3}%
\crefformat{figure}{#2Figure~#1#3}%
\Crefformat{figure}{#2Figure~#1#3}%
\crefformat{section}{#2Section~#1#3}
\Crefformat{section}{#2Section~#1#3}
\crefformat{chapter}{#2Chapter~#1#3}
\Crefformat{chapter}{#2Chapter~#1#3}
\crefformat{chapter*}{#2Chapter~#1#3}
\Crefformat{chapter*}{#2Chapter~#1#3}
\crefformat{part}{#2Part~#1#3}
\Crefformat{part}{#2Part~#1#3}
\crefformat{enumi}{#2(#1)#3}
\Crefformat{enumi}{#2(#1)#3}

\usepackage{enumerate}

\usepackage{latexsym}

\theoremnumbering{arabic}
\theoremstyle{plain}
\theoremsymbol{}
\theorembodyfont{\itshape}
\theoremheaderfont{\normalfont\bfseries}
\theoremseparator{.}

\newtheorem{theorem}{Theorem}
\crefformat{theorem}{#2Theorem~#1#3}
\Crefformat{theorem}{#2Theorem~#1#3}

\newcommand{\newtheoremwithcrefformat}[2]{%
  \newtheorem{#1}[theorem]{#2}%
  \crefformat{#1}{##2\MakeUppercase#1~##1##3}%
  \Crefformat{#1}{##2\MakeUppercase#1~##1##3}%
}
\newcommand{\newseptheoremwithcrefformat}[2]{%
  \newtheorem{#1}{#2}%
  \crefformat{#1}{##2\MakeUppercase#1~##1##3}%
  \Crefformat{#1}{##2\MakeUppercase#1~##1##3}%
}

\newtheoremwithcrefformat{lemma}{Lemma}
\newtheoremwithcrefformat{proposition}{Proposition}
\newtheoremwithcrefformat{observation}{Observation}
\newtheoremwithcrefformat{conjecture}{Conjecture}
\newtheoremwithcrefformat{corollary}{Corollary}
\newseptheoremwithcrefformat{claim}{Claim}
\theorembodyfont{\upshape}
\newtheoremwithcrefformat{example}{Example}
\newtheoremwithcrefformat{remark}{Remark}
\newtheoremwithcrefformat{definition}{Definition}

\theoremstyle{nonumberplain}
\theoremheaderfont{\scshape}
\theorembodyfont{\normalfont}
\theoremsymbol{\ensuremath{\square}}
\newtheorem{proof}{Proof}

\theoremsymbol{\ensuremath{\lrcorner}}
\newtheorem{clproof}{Proof}

\def\cqedsymbol{\ifmmode$\lrcorner$\else{\unskip\nobreak\hfil
\penalty50\hskip1em\null\nobreak\hfil$\lrcorner$
\parfillskip=0pt\finalhyphendemerits=0\endgraf}\fi}

\newcommand{\Oh}{\mathcal{O}}

\newcommand{\Rr}{\mathcal{R}}

\newcommand{\eps}{\varepsilon}

\newcommand{\R}{\mathbb{R}}

\renewcommand{\phi}{\varphi}
\renewcommand{\epsilon}{\varepsilon}

\newcommand{\dist}{\mathrm{dist}}

\newcommand{\VorPrt}{\mathsf{Cell}}
\newcommand{\VorTree}{T}
\newcommand{\VorDiag}{\mathsf{Vor}}
\newcommand{\Spoke}{\mathsf{Spoke}}
\newcommand{\Diam}{\mathsf{Diam}}

\newcommand{\prt}{\partial}

\newcommand{\Ii}{\mathcal{I}}

\newcommand{\openfac}{D}
\newcommand{\fac}{F}
\newcommand{\clients}{C}
\newcommand{\cost}{\mathsf{cost}}
\newcommand{\conncost}{\mathsf{conn}}
\newcommand{\opencost}{\mathsf{open}}
\newcommand{\wei}{\omega}
\newcommand{\supp}{\mathrm{supp}}
\newcommand{\lvl}{\mathsf{level}}

\newcommand{\apxfac}{{\widetilde{\openfac}}}
\newcommand{\optfac}{\openfac^\star}
\newcommand{\inst}{\mathcal{I}}

\newcommand{\maxlvl}{L}
\newcommand{\port}[1]{\mathfrak{p}\langle #1\rangle}
\newcommand{\portal}{\mathfrak{p}}

\renewcommand{\leq}{\leqslant}
\renewcommand{\geq}{\geqslant}

\begin{document}

\title{Efficient approximation schemes for uniform-cost clustering problems in planar graphs\thanks{This work is 
a part of projects CUTACOMBS (Ma. Pilipczuk) and TOTAL (Mi. Pilipczuk) that have received funding from the European Research Council (ERC) 
under the European Union's Horizon 2020 research and innovation programme (grant agreements No.~714704 and No.~677651, respectively).}}

\author{
Vincent Cohen-Addad\thanks{
 Sorbonne Universit\'e, CNRS, Laboratoire d'informatique de Paris 6, LIP6, F-75252 Paris, France \texttt{vincent.cohen.addad@ens-lyon.org}.
}
\and
Marcin~Pilipczuk\thanks{
  Institute of Informatics, University of Warsaw, Poland, \texttt{marcin.pilipczuk@mimuw.edu.pl}.
}
\and
Micha\l{}~Pilipczuk\thanks{
  Institute of Informatics, University of Warsaw, Poland, \texttt{michal.pilipczuk@mimuw.edu.pl}.
}
}

\begin{titlepage}
\def\thepage{}
\thispagestyle{empty}
\maketitle

\begin{textblock}{20}(0, 12.5)
\includegraphics[width=40px]{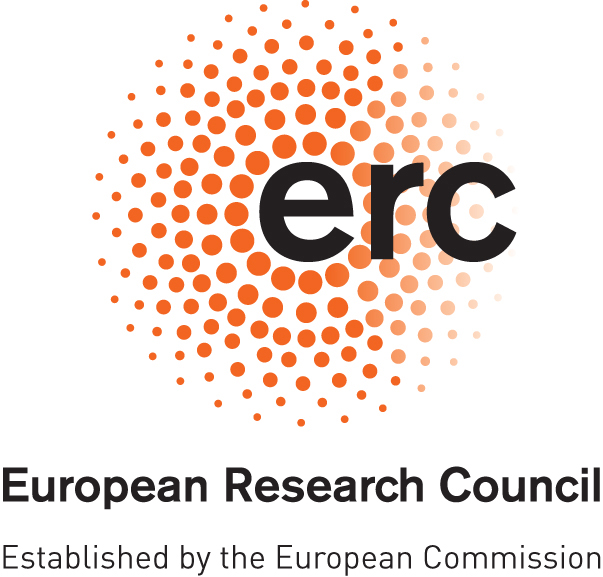}%
\end{textblock}
\begin{textblock}{20}(-0.25, 12.9)
\includegraphics[width=60px]{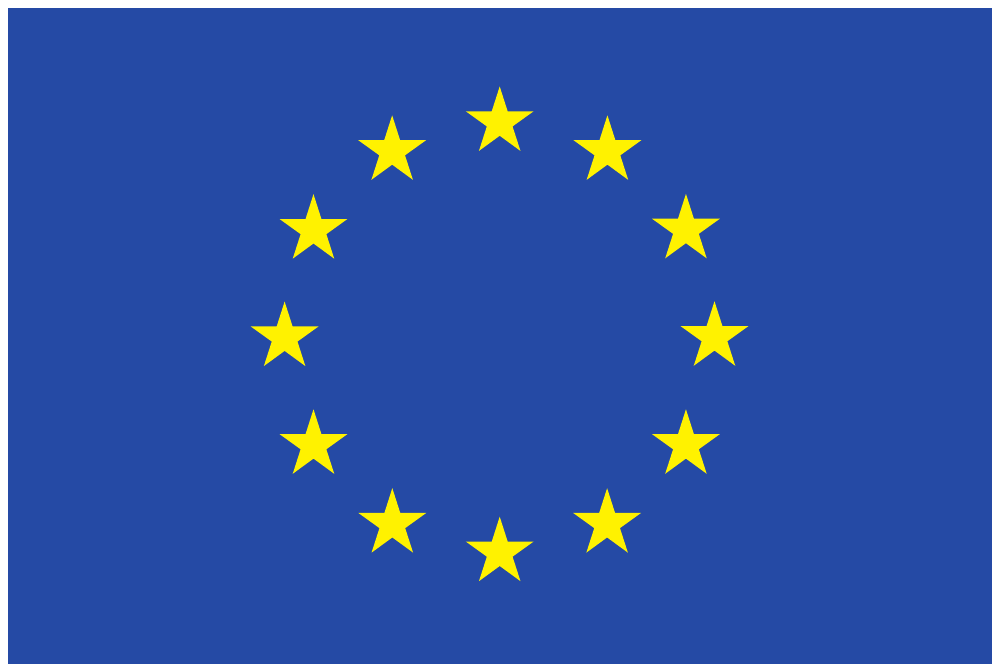}%
\end{textblock}

\begin{abstract}
We consider the {\sc{$k$-Median}} problem on planar graphs: 
given an edge-weighted planar graph $G$, a set of clients $\clients\subseteq V(G)$, a set of facilities $\fac\subseteq V(G)$, and an integer parameter $k$,
the task is to find a set of at most $k$ facilities whose opening minimizes the total connection cost of clients, 
where each client contributes to the cost with the distance to the closest open facility.
We give two new approximation schemes for this problem:
\begin{itemize}
\item {\em{FPT Approximation Scheme}}: for any $\eps>0$, in time $2^{\Oh(k\eps^{-3}\log (k\eps^{-1}))}\cdot n^{\Oh(1)}$ we can compute a solution that 
      has connection cost at most $(1+\eps)$ times the optimum, with high probability.
\item {\em{Efficient Bicriteria Approximation Scheme}}: for any $\eps>0$, in time $2^{\Oh(\eps^{-5}\log (\varepsilon^{-1}))}\cdot n^{\Oh(1)}$ we can compute a set of at most $(1+\eps)k$ facilities 
      whose opening yields connection cost at most $(1+\eps)$ times the optimum connection cost for opening at most $k$ facilities, with high probability.
\end{itemize}
As a direct corollary of the second result we obtain an EPTAS for {\sc{Uniform Facility Location}} on planar graphs, with same running time.

Our main technical tool is a new construction of a ``coreset for facilities'' for {\sc{$k$-Median}} in planar graphs: 
we show that in polynomial time one can compute a subset of facilities $\fac_0\subseteq \fac$ of size $k\cdot (\log n/\eps)^{\Oh(\eps^{-3})}$ with a guarantee that there is a $(1+\eps)$-approximate
solution contained in $\fac_0$.
\end{abstract}

\end{titlepage}




\section{Introduction}
We study approximation schemes for classic clustering objectives, formalized as follows. Given an edge-weighted graph $G$ together
with a set $\clients$ of vertices called \emph{clients}, a set $\fac$ of vertices called \emph{candidate facilities}, and an \emph{opening cost} $\opencost \in \R_{\geq 0}$,
the {\sc{Uniform Facility Location}} problem asks for a subset of facilities (also called centers) $\openfac \subseteq \fac$ that minimizes the cost defined as
$|\openfac| \cdot \opencost + \sum_{c \in \clients} \min_{f \in \openfac} \dist(c,f)$. 
In the \textsc{Non-uniform Facility Location} variant, the opening costs may vary between facilities.

We also consider the related \textsc{$k$-Median} problem, where the tuple $(G, \clients, \fac)$ comes with a hard
budget $k$ for the number of open facilities (as opposed to the opening cost $\opencost$).
That is, the problem asks
for a set $\openfac \subseteq \fac$ of size at most $k$ that minimizes the connection cost $\sum_{c \in \clients} \dist(c, \openfac)$.
Note that {\sc{Uniform Facility Location}} can be reduced to \textsc{$k$-Median} by guessing the number of open facilities
in an optimal solution.

{\sc{Facility Location}} and {\sc{$k$-Median}} model in an abstract way various clustering objectives appearing in applications. 
Therefore, designing approximation algorithms for them and their variants is a vibrant topic in the field of approximation algorithms. 
For {\sc{Non-uniform Facility Location}}, a long line of work~\cite{AGKMMP04,Hochbaum82,STA1997,JaV01} culminated with the $1.488$-approximation algorithm by Li~\cite{Li13}.
On the other hand, Guha and Khuller~\cite{GuK99} showed that the problem cannot be approximated in polynomial time within factor better than $1.463$ unless $\mathsf{NP}\subseteq \mathsf{DTIME}[n^{\Oh(\log \log n)}]$, 
which gives almost tight bounds on the best approximation factor achievable in polynomial time.
For {\sc{$k$-Median}}, the best known approximation ratio achievable in polynomial time is $2.67$ due to Byrka et al.~\cite{ByrkaPRST17}, while the lower bound of $1.463$ due to
Guha and Khuller~\cite{GuK99} holds here as well.

Given the approximation hardness status presented above, it is natural to consider restricted metrics.
In this work we consider {\em{planar metrics}}: we assume that the underlying edge-weighted graph $G$ is planar.

It was a long-standing open problem whether {\sc{Facility Location}} admits a polynomial-time approximation scheme (PTAS) in planar metrics.
For the uniform case, this question has been resolved in affirmative by Cohen-Addad et al.~\cite{local-search} in an elegant way:
they showed that local search of radius $\Oh(1/\eps^2)$ actually yields a $(1+\eps)$-approximation, giving a PTAS with running time
$n^{\Oh(1/\eps^2)}$.
This approach also gives a PTAS for {\sc{$k$-Median}} with a similar running time, and works even in metrics induced by graphs from any fixed proper minor-closed class.

Very recently, Cohen-Addad et al.~\cite{CohenPP19} also gave a PTAS for {\sc{Non-uniform Facility Location}} in planar metrics using a different approach.
Roughly, the idea is to first apply Baker layering scheme to reduce the problem to the case when in all clusters (sets of clients connected to the same facility) in the solution, 
all clients are within distance between $1$ and~$r$ from the center, for some constant $r$ depending only on $\eps$.
This case is then resolved by another application of Baker layering scheme, followed by a dynamic programming on a hierarchichal decomposition of the graph using shortest paths as balanced separators.

Both the schemes of~\cite{local-search} and of~\cite{CohenPP19} are PTASes: they run in time $n^{g(\eps)}$ for some function~$g$.
It is therefore natural to ask for an {\em{efficient PTAS}} ({\em{EPTAS}}): an approximation scheme with running time $f(\eps)\cdot n^{\Oh(1)}$ for some function $f$.
Recently, such an EPTAS was given by Cohen-Addad~\cite{Cohen-Addad:2018} for {\sc{$k$-Means}} in low-dimensional Euclidean spaces; 
this is a variant of {\sc{$k$-Median}} where every client contributes to the connection cost with the {\em{square}} of its distance from the closest open facility.
Here, the idea is to apply local search as in~\cite{local-search}, but to use the properties of the metric to explore the local neighborhood faster.
Unfortunately, this technique mainly relies on the Euclidean structure (or on the bounded doubling dimension of the input) and seems hard to lift to the general planar case.
Also the techniques of~\cite{CohenPP19} are far from yielding an EPTAS: essentially, one needs to use a logarithmic number of portals at every step of the final dynamic programming in order to tame
the accumulation of error through $\log n$ levels of the decomposition.

The goal of this work is to circumvent these difficulties and give an EPTAS for {\sc{Uniform Facility Location}} in planar metrics.

\subparagraph*{Our results.}
Our main technical contribution is the following theorem.
In essence, it states that when solving {\sc{$k$-Median}} on a planar graph one can restrict the facility set to a subset of size $k\cdot (\varepsilon^{-1} \log n)^{\Oh(\varepsilon^{-3})}$,
at the cost of losing a multiplicative factor of $(1+\eps)$ on the optimum connection cost.
This can be seen as the planar version of the classic result by Matou\v{s}ek~\cite{Mat00} who
showed that for Euclidean metrics of dimension $d$, it is possible to reduce the number
of candidate centers to $\text{poly}(k) \eps^{-O(d)}$
at the cost of losing a multiplicative factor of $(1+\eps)$ on the optimum connection cost
(through the use
of coresets as well). For general metrics, obtaining such a result seems challenging, since this would imply
a $(1+\eps)$-approximation algorithm with running time $f(k,\eps) n^{O(1)}$, which would contradict
Gap-ETH~\cite{absFPTkmed}.

From now on, by {\em{with constant probability}} we mean with probability at least~$1/2$; this can be boosted by independent repetition.

\begin{theorem}\label{thm:fac-coreset}
Given a {\sc{$k$-Median}} instance $(G,\fac,\clients,k)$, where $G$ is a planar graph,
and an accuracy parameter $\varepsilon > 0$, 
one can in randomized polynomial time compute a set $\fac_0 \subseteq \fac$ of size
$k \cdot (\varepsilon^{-1} \log n)^{\Oh(\varepsilon^{-3})}$ satisfying the following condition with constant probability: there exists a set
$\openfac_0 \subseteq \fac_0$ of size at most $k$ such that for every 
set $\openfac \subseteq \fac$ of size at most $k$ it holds that
$\conncost(\openfac_0, \clients) \leq (1+\varepsilon)\cdot \conncost(\openfac, \clients)$.
\end{theorem}

A direct corollary of Theorem~\ref{thm:fac-coreset} is
a fixed-parameter approximation scheme for the \textsc{$k$-Median} problem
in planar graphs.
This continues the line of work on fixed-parameter approximation schemes 
for $k$-median and $k$-means in Euclidean spaces~\cite{DBLP:conf/stoc/VegaKKR03,KSS10}, where
the goal is to design an algorithm running in time
$f(k,\eps)\cdot n^{\Oh(1)}$ for a computable function $f$.

\begin{theorem}\label{thm:kmedian-fpas}
Given a {\sc{$k$-Median}} instance $(G,\clients,\fac,k)$, where $G$ is a planar graph,
and an accuracy parameter $\varepsilon > 0$, one can in randomized time
$2^{\Oh(k\varepsilon^{-3}\log (k\varepsilon^{-1}))} \cdot n^{\Oh(1)}$ compute a solution $\openfac \subseteq \fac$ 
that has connection cost at most $(1+\varepsilon)$ times the minimum possible connection cost
with constant probability.
\end{theorem}
\begin{proof}
Apply the algorithm of Theorem~\ref{thm:fac-coreset} and let $\fac_0\subseteq \fac$ be the obtained subset of facilities.
Then run a brute-force search through all subsets of $\fac_0$ of size at most $k$ and output one with the smallest connection cost. Thus, the running time is
$$\left(k \cdot (\varepsilon^{-1} \log n)^{\Oh(\varepsilon^{-3})}\right)^k\cdot n^{\Oh(1)}\leq 
2^{\Oh(k\eps^{-3}\log (k\eps^{-1}))}\cdot (\log n)^{\Oh(k\eps^{-3})}\cdot n^{\Oh(1)}\leq 2^{\Oh(k\eps^{-3}\log (k\eps^{-1}))}\cdot n^{\Oh(1)},$$
where the last inequality follows from the bound $(\log n)^d\leq 2^{\Oh(d\log d)}\cdot n^{\Oh(1)}$, which can be proved as follows:
if $n\leq 2^{d^2}$ then $(\log n)^d\leq d^{2d}\leq 2^{\Oh(d\log d)}$, and if $n>2^{d^2}$ then $(\log n)^d\leq 2^{\sqrt{\log n}\cdot \log\log n}\leq n^{\Oh(1)}$.
\end{proof}

Using Theorem~\ref{thm:fac-coreset} we can also give an efficient bicriteria PTAS for \textsc{$k$-Median} in planar graphs. 
This time, the proof is more involved and uses the local search techniques of~\cite{Cohen-Addad:2018}.

\begin{theorem}\label{thm:kmedian-bi}
Given a \textsc{$k$-Median} instance $(G,\clients,\fac,k)$, where $G$ is a planar graph,
and an accuracy parameter $\varepsilon > 0$, one can in randomized time
$2^{\Oh(\varepsilon^{-5} \log(\varepsilon^{-1}))} \cdot n^{\Oh(1)}$ compute a set $\openfac \subseteq \fac$ of size
at most $(1+\varepsilon)k$ such that its connection cost is at most $(1+\varepsilon)$ times the minimum
possible connection cost for solutions of size $k$ with constant probability.
\end{theorem}

A direct corollary of Theorem~\ref{thm:kmedian-bi} is an efficient PTAS
for \textsc{Uniform Facility Location} in planar graphs.

\begin{theorem}\label{thm:ufl}
Given a {\sc{Uniform Facility Location}} instance $(G,\clients,\fac,\opencost)$, where $G$
is a planar graph, and an accuracy parameter $\varepsilon > 0$, one can in randomized time
$2^{\Oh(\eps^{-5}\log (\varepsilon^{-1}))}\cdot n^{\Oh(1)}$ compute a solution $\openfac \subseteq \fac$
that has total cost at most $(1+\varepsilon)$ times the optimum cost with constant probability.
\end{theorem}
\begin{proof}
Iterate over all possible choices of $k$ being the number of facilities opened by the optimum
solution, and for every $k$ invoke the algorithm of Theorem~\ref{thm:kmedian-bi} 
for the \textsc{$k$-Median} instance $(G,\clients,\fac,k)$. 
From the obtained solutions output one with the smallest cost.
\end{proof}
Note that the approach presented above fails for the non-uniform case,
 where each facility has its own, distinct opening cost.

In this extended abstract we focus on proving the main result, Theorem~\ref{thm:fac-coreset}. The proof of Theorem~\ref{thm:kmedian-bi}, on which Theorem~\ref{thm:ufl} also relies,
is deferred to Section~\ref{sec:kmedian-bi}.

\subparagraph*{Our techniques.}
The first step in the proof of Theorem~\ref{thm:fac-coreset} is to reduce the number of relevant clients 
using the coreset construction of Feldman and Langberg~\cite{coreset}. By applying this technique, 
we may assume that there are at most $k\cdot \Oh(\eps^{-2}\log n)$ clients in the instance, however they are weighted:
every client $c$ is assigned a nonnegative weight $\wei(c)$, and it contributes to the connection cost of any solution with $\wei(c)$ times the distance to the closest open facility in the solution.

We now examine the Voronoi diagram induced in the input graph $G$ by the {\em{clients}}: vertices of $G$ are classified into {\em{cells}} according to the closest client.
This Voronoi diagram has one cell per every client, thus it can be regarded as a planar graph with $|\clients|$ faces, where each face accommodates one cell.
To formally define the Voronoi diagram, and in particular the boundaries between neighboring cells, we use the framework introduced by 
Marx and Pilipczuk~\cite{MarxP15} and its extension used in~\cite{PilipczukLW18}.

Consider now all the {\em{spokes}} in the diagram, where a spoke is the shortest path connecting the center of a cell (i.e. a client) with a branching node of the diagram incident to the cell (which is a face of $G$).
Removing all the spokes and all the branching nodes from the plane divides it into {\em{diamonds}}, where each diamond is delimited by four spokes, called further the {\em{perimeter}} of the diamond.
See Figure~\ref{fig:voronoi} for an example. Since the diagram is a planar graph with $|\clients|$ faces, there are $\Oh(|\clients|)=k\cdot \Oh(\eps^{-2}\log n)$ diamonds altogether. 
Moreover, since no diamond contains a client in its interior, whenever $P$ is a path connecting a client with a facility belonging to some diamond $\Delta$, $P$ has to cross the perimeter of $\Delta$.

Now comes the key and most technical part of the proof.
We very carefully put $\Oh(\eps^{-2}\log n)$ portals on the perimeter of each diamond.
The idea of placement is similar to that of the resolution metric used in the QPTAS for {\sc{Facility Location}}.
Namely, on a spoke $Q$ starting at client $c$ we put portals at distance $1,(1+\eps),(1+\eps)^2,\ldots$ from $c$, so that the further we are on the spoke from $c$, the sparser the portals are.
As a diamond is delimited by four spokes, we may thus use only $\Oh(\eps^{-2}\log n)$ portals per diamond, while the cost of snapping a path crossing $Q$ to the portal closest to the crossing point
can be bounded by $\eps$ times the distance from the crossing point to $c$.

For a facility $f$ in a diamond $\Delta$, we define the \emph{profile} of $f$
as follows. For every spoke $Q$ in the perimeter of $\Delta$, we look at the closest portal $\pi$ from $f$
on $Q$. We record approximate (up to $(1+\eps)$ multiplicative error) distances
from $f$ to $\pi$ and $\Oh(\eps^{-3})$ neighboring portals, as well as the distance to the
client endpoint of the spoke $Q$.
The crux lies in the following fact:
for every two facilities $f,f'$ in $\Delta$ with the same profile,
replacing $f$ with $f'$ increases the connection cost of any client $c$ connected to $f$ only by a multiplicative factor of $(1+\eps)$. 
Hence, for every profile in every diamond it suffices to keep just one facility with that profile.
Since there are $k\cdot \Oh(\eps^{-2}\log n)$ diamonds and $\Oh(\eps^{-1} \log n)^{\Oh(\eps^{-3})}$ possible profiles in each of them, we keep at most $k \cdot (\varepsilon^{-1} \log n)^{\Oh(\varepsilon^{-3})}$
facilities in total. This proves Theorem~\ref{thm:fac-coreset}.

For the proof of Theorem~\ref{thm:kmedian-bi}, we first apply Theorem~\ref{thm:fac-coreset} to reduce the number of facilities to $k \cdot (\varepsilon^{-1} \log n)^{\Oh(\varepsilon^{-3})}$. 
Then we again inspect the Voronoi diagram, but now induced by the {\em{facilities}}. Having contracted every cell to a single vertex, we compute an $r$-division of the obtained planar graph to cover it with 
regions of size $r=(\eps^{-1}\log n)^{\Oh(\varepsilon^{-3})}$ so that only $\Oh(\eps) k$ facilities are on boundaries of the regions. We open all the facilities in all the boundaries --- thus exceeding the quota 
for open facilities by $\Oh(\eps) k$ --- run the PTAS of Cohen-Addad et al.~\cite{local-search} in each region independently, and at the end assemble regional solutions using a knapsack dynamic programming. 
Since within each region there are only polylogarithmically many facilities, each application of the PTAS actually works in time $f(\eps)\cdot n^{\Oh(1)}$.


\section{Preliminaries on Voronoi diagrams and coresets}\label{sec:prelims}

In this section we recall some tools about Voronoi diagrams in planar graphs and coresets that will be used in the proof of Theorem~\ref{thm:fac-coreset}.
We will consider undirected graphs with positive edge lengths embedded in a sphere,
with the standard shortest-paths metric $\dist(u,v)$ for $u,v\in V(G)$.
Contrary to the previous section, the metric is defined on the vertex set of $G$ only, i.e., we do not consider $G$ as a metric space with points in the interiors of edges.
For $X, Y \subseteq V(G)$, we denote $\dist(X,Y) = \min_{x \in X, y \in Y} \dist(x, y)$ and similarly we define $\dist(u, X)$ for $u \in V(G)$ and $X \subseteq V(G)$.

Recall that for a set $\openfac \subseteq V(G)$ of \emph{open facilities}
and a set $\clients \subseteq V(G)$,
    we define the \emph{connection cost} as 
\begin{equation*}\conncost(\openfac, \clients) = \sum_{v \in \clients} \dist(v, \openfac).\end{equation*}
If the input is additionally equipped with opening costs $\opencost \colon \fac \to \mathbb{R}_{\geq 0}$, then the \emph{opening cost} of $\openfac$ is defined as $\sum_{w \in \openfac} \opencost(w)$.

\subsection{Voronoi diagrams in planar graphs}\label{sec:voronoi}

We now recall the construction of Voronoi diagrams and related notions in planar graphs used by Marx and Pilipczuk~\cite{MarxP15}.
The setting is as follows. Suppose $G$ is an $n$-vertex simple graph embedded in a sphere $\Sigma$ whose edges are assigned nonnegative real lengths.
We consider the shortest path metric in $G$: for two vertices $u,v$, their {\em{distance}} $\dist(u,v)$ is equal to the smallest possible total length of a path from $u$ to $v$.
We will assume that $G$ is triangulated (i.e. every face of $G$ is a cycle of length $3$), for this may always be achieved by triangulating the graph using edges of infinite weight.

Further, we assume that shortest paths are unique in $G$ and that finite distances between distinct vertices in $G$ are pairwise different: 
for all vertices $u,v,u',v'$ with $u\neq v$, $u'\neq v'$ and $\{u,v\}\neq \{u',v'\}$, we have $\dist(u,v)\neq \dist(u',v')$ or $\dist(u,v)=\dist(u',v')=+\infty$.
This can be achieved by adding small perturbations to the edge lengths. 
Since we never specify degrees of polynomials in the running time of our algorithms, we may ignore the additional complexity cost incurred by the need of handling the perturbations in arithmetic operations.

\subparagraph*{Voronoi diagrams and their properties.}
Suppose that $S$ is a subset of vertices\footnote{In~\cite{MarxP15} a more general setting is considered 
where objects inducing the diagram are connected subgraphs of $G$ instead of single vertices. We will not need this generality here.} of $G$. 
First, define the {\em{Voronoi partition}}: for a vertex $p\in S$, the {\em{Voronoi cell}} $\VorPrt_S(p)$ is the set of all those vertices $u\in V(G)$ whose distance from $p$ is smaller
than the distance from any other vertex $q\in S$; note that ties do not occur due to the distinctness of distances in $G$.
Note that $\{\VorPrt_S(p)\}_{p\in S}$ is a partition of the vertex set of $G$.
For each $p\in S$, let $\VorTree(p)$ be the union of shortest paths from vertices of $\VorPrt_S(p)$ to $p$; recall here that shortest paths in $G$ are unique.
Note that, due to the distinctness of distances in $G$, $\VorTree(p)$ is a spanning tree of the subgraph of $G$ induced by the cell $\VorPrt_S(p)$.

The diagram $\VorDiag_S$ induced by $G$ is a multigraph constructed as follows.
First, take the dual $G^\star$ of $G$ and remove all edges dual to the edges of all the trees $\VorTree(p)$, for $p\in S$.
Then, exhaustively remove vertices of degree $1$.
Finally, for every maximal $2$-path (i.e. path with internal vertices of degree $2$), say with endpoints $u$ and $v$, we replace this path by the edge $uv$; note that this creates a loop at $u$ in case $u=v$.
The resulting multigraph $\VorDiag_S$ is the Voronoi diagram induced by $S$.
Note that the vertices of $\VorDiag_S$ are faces of $G$; for clarity we shall call them {\em{branching nodes}}.
Furthermore, $\VorDiag_S$ inherits an embedding in $\Sigma$ from the dual $G^\star$, where an edge $uv$ that replaced a maximal $2$-path $P$ is embedded precisely as $P$, 
i.e., as the concatenation of (the embeddings of) the edges comprising $P$. From now on we will assume this embedding of $\VorDiag_S$.

We recall several properties of $\VorDiag_S$, observed in~\cite{MarxP15}
\begin{lemma}[Lemmas~4.4 and~4.5 of~\cite{MarxP15}]
The diagram $\VorDiag_S$ is a connected and $3$-regular multigraph embedded in $\Sigma$, which has exactly $|S|$ faces, $2|S|-4$ branching nodes, and $3|S|-6$ edges.
The faces of $\VorDiag_S$ are in one-to-one correspondence with vertices of $S$: each $p\in S$ corresponds to a face of $\VorDiag_S$ that contains all vertices of $\VorPrt_S(p)$ and 
no other vertex of $G$.
\end{lemma}

\subparagraph*{Spokes and diamonds.} We now introduce further structural elements that can be distinguished in the Voronoi diagram, see Figure~\ref{fig:voronoi} for reference.
The definitions and basic observations presented below are taken from Pilipczuk et al.~\cite{PilipczukLW18}, and were inspired by the Euclidean analogues due to Har-Peled~\cite{Har-Peled14}.

An {\em{incidence}} is a triple $\tau=(p,u,f)$ where $p\in S$, $f$ is a branching node of the diagram $\VorDiag_S$, and $u$ is a vertex of $G$ 
that lies on $f$ (recall that $f$ is a triangular face of $G$) and belongs to $\VorPrt_S(p)$. 
The {\em{spoke}} of the incidence $\tau$, denoted $\Spoke(\tau)$, is the shortest path in $G$ between $p$ and $u$.
Note that all the vertices of $\Spoke(\tau)$ belong to $\VorPrt_S(p)$.

\begin{figure}[t]
  \centering
  \includegraphics[width=0.6\textwidth]{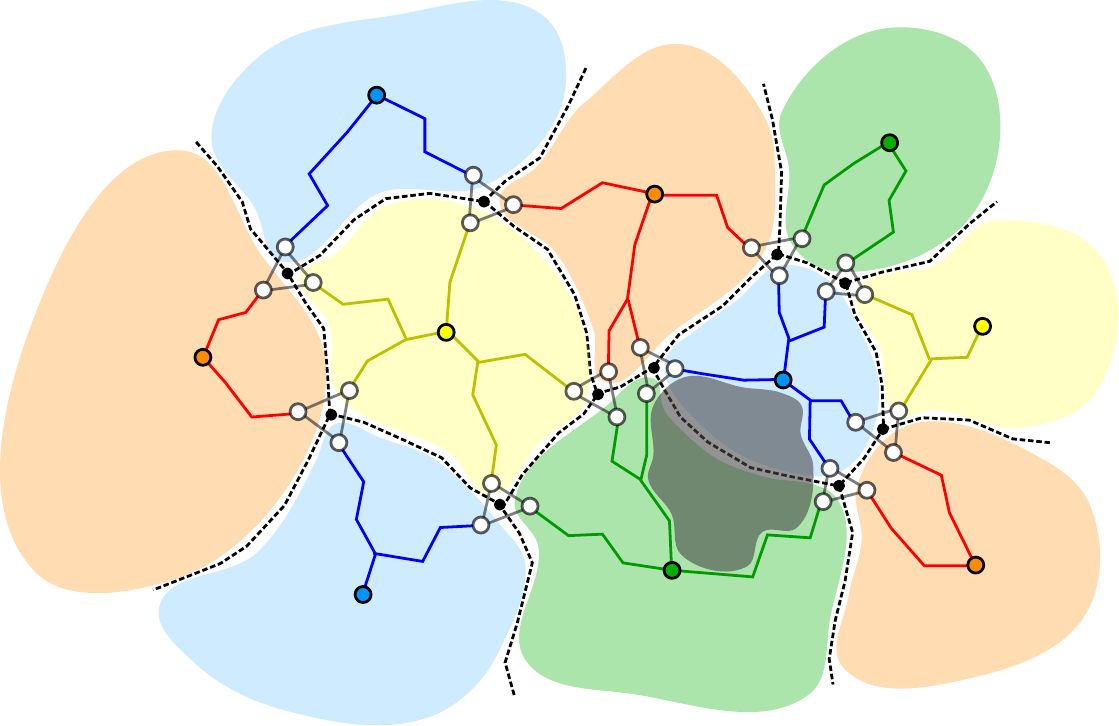}
  \caption{A part of the Voronoi diagram with various features distinguished. Branching nodes of the diagram are grayed triangular faces, edges of the diagram are dashed. 
           Solid paths of respective colors are spokes. (The interior of) one diamond is grayed in order to highlight it.}\label{fig:voronoi}
\end{figure}

Let $e=f_1f_2$ be an edge of the diagram $\VorDiag_S$, where $f_1,f_2$ are branching nodes (possibly $f_1=f_2$ if $e$ is a loop in $\VorDiag_S$).
Further, let $p_1$ and $p_2$ be the vertices from $S$ that correspond to faces of $\VorDiag_S$ incident to $e$ (possibly $p_1=p_2$ if $e$ is a bridge in $\VorDiag_S$).
Suppose for a moment that $f_1\neq f_2$. Then, out of the three edges of $f_1$ (these are edges in $G$) there is exactly one that crosses the edge $e$ of $\VorDiag_S$; say it is the edge
$u_{1,1}u_{1,2}$ where $u_{1,1}\in \VorPrt_S(p_1)$ and $u_{1,2}\in \VorPrt_S(p_2)$. Symmetrically, there is one edge of $f_2$ that crosses $e$, 
say it is $u_{2,1}u_{2,2}$ where $u_{2,1}\in \VorPrt_S(p_1)$ and $u_{2,2}\in \VorPrt_S(p_2)$. In case $f_1=f_2$, the edge $e$ crosses two different edges of $f_1=f_2$ and we define $u_{1,1},u_{1,2},u_{2,1},u_{2,2}$
analogously for these two crossings; note that then, provided $p_1$ corresponds to the face enclosed by the loop $e$, we have $u_{1,1}=u_{1,2}$.
For all $i,j\in \{1,2\}$, consider the incidence $\tau_{i,j}=(p_i,u_{i,j},f_j)$.

Consider removing the following subsets from the sphere $\Sigma$: interiors of faces $f_1,f_2$ and spokes $\Spoke(\tau_{i,j})$ for all $i,j\in \{1,2\}$.
After this removal the sphere breaks into two regions, out of which exactly one, say $R$, intersects (the embedding of) $e$. 
Let the {\em{diamond}} of $e$, denoted $\Diam(e)$, be the subgraph of $G$ consisting of all features (vertices and edges) embedded in $R\cup \bigcup_{i,j\in \{1,2\}} \Spoke(\tau_{i,j})$.
The region $R$ as above is the {\em{interior}} of the diamond $\Diam(e)$.
Note that in particular, the spokes $\Spoke(\tau_{i,j})$ for $i,j\in \{1,2\}$ and the edges $u_{1,1}u_{1,2}$ and $u_{2,1}u_{2,2}$ belong to $\Diam(e)$.
The {\em{perimeter}} of the diamond of $e$ is the closed walk obtained by concatenating spokes $\Spoke(\tau_{i,j})$ for $i,j\in \{1,2\}$ and edges $u_{1,1}u_{1,2},u_{2,1}u_{2,2}$
in the natural order around $\Diam(e)$. The following observation is immediate:

\begin{proposition}\label{prop:diagram}
Consider removing all the spokes (considered as curves on $\Sigma$) and all the branching nodes (considered as interiors of faces on $\Sigma$) of the diagram $\VorDiag_S$ from the sphere $\Sigma$.
Then $\Sigma$ breaks into $3|S|-6$ regions that are in one-to-one correspondence with edges of $\VorDiag_S$: a region corresponding to the edge $e$ is the interior of the diamond $\Diam(e)$.
Consequently, the intersection of diamonds of two different edges of $\VorDiag_S$ is contained in the intersection of their perimeters.
\end{proposition}

Finally, we note that the perimeter of a diamond separates it from the rest of the graph. Since vertices of $S$ are never contained in the interior of a diamond, this yields the following.

\begin{lemma}\label{lem:perimeter-intersects}
Let $p\in S$ and $u$ be a vertex of $G$ belonging to the diamond $\Diam(e)$ for some edge $e$ of $\VorDiag_S$. Then every path in $G$ connecting $u$ and $p$ intersects the perimeter of $\Diam(e)$.
\end{lemma}

\subsection{Coresets}

In most our algorithms, the starting point is the notion of a \emph{coreset} and a corresponding result of Feldman and Langberg~\cite{coreset}. 
To this end, we need to slightly generalize the notion of a client set in a \textsc{$k$-Median} instance.
A \emph{client weight function} is a function $\wei \colon \clients \to \mathbb{R}_{\geq 0}$.
Given a set $\openfac \subseteq \fac$ of open facilities, the (weighted) connection cost is defined as
\begin{equation*}\conncost(\openfac, \wei) = \sum_{v \in \clients} \dist(v, \openfac) \cdot \wei(v).\end{equation*}
That is, every client $v$ is assigned a weight $\wei(v)$ with which it contributes to the objective function.
The \emph{support} of a weight function $\wei$ is defined as $\supp(\wei) = \{v \in \clients~|~\wei(v) > 0\}$. 
From now on, whenever we speak about a \textsc{$k$-Median} instance without specified client weight function, we assume the standard function assigning each client weight $1$.

The essence of coresets is that one can find weight functions with small support that well approximate the original instance.
Given a \textsc{$k$-Median} instance $(G,\fac,\clients,k)$ (without weights) and an accuracy parameter $\varepsilon > 0$, a \emph{coreset} is a weight function $\wei$ such that for every 
set $\openfac \subseteq \fac$ of size at most $k$, it holds that
\begin{equation*}|\conncost(\openfac, \clients) - \conncost(\openfac, \wei)| \leq \varepsilon\cdot \conncost(\openfac, \clients).\end{equation*}
We rely on the following result of Feldman and Langberg~\cite{coreset}.

\begin{theorem}[Theorem 15.4 of \cite{coreset}]\label{thm:coreset}
Given a \textsc{$k$-Median} instance $(G,\fac,\clients,k)$ with $n = |V(G)|$ and accuracy parameter $\eps>0$,
one can in randomized polynomial time 
find a weight function $\wei$ with support of size $\Oh(k\eps^{-2} \log n)$
that is a coreset with constant probability.
\end{theorem}

We note that Ke Chen~\cite{chen-coreset} gave a construction of a strong coreset with support of size $\Oh(k^2\eps^{-2}\log n)$ that is much simpler than the later construction of Feldman and Langberg~\cite{coreset}.
By using this construction instead, we would obtain a weaker version of Theorem~\ref{thm:fac-coreset}, with a bound on $|\fac_0|$ that is quadratic in $k$ instead of linear.
This would be perfectly sufficient to derive an FPT approximation scheme as in Theorem~\ref{thm:kmedian-fpas}, but for Theorem~\ref{thm:kmedian-bi} we will vitally use the stronger statement.
A construction of coresets with similar size guarantees, but maintainable in the streaming model, has been proposed by Braverman et al.~\cite{BravermanFL16}.

\subparagraph*{Divisions.} 
A {\em{division}} of graph $G$ is a family $\Rr$ of subgraphs of $G$, called {\em{regions}},
such that every edge of $G$ is contained in exactly one region and every vertex of $G$ is contained in at least one region.
For a region $R\in \Rr$, the {\em{boundary}} of $R$, denoted $\prt R$, is the set of those vertices of $R$ that are contained also in some other region $R'\in \Rr$.
For a positive integer $r$, a division $\Rr$ is called an {\em{$r$-division}} if every region contains at most $r$ vertices.

The following lemma, which can be traced to the work of Fredrickson~\cite{Frederickson87}, expresses the well-known property that planar graphs admit $r$-divisions with small boundary.
We remark that Fredrickson gave $r$-divisions with stronger properties, but this will be the generality needed here.

\begin{lemma}[\cite{Frederickson87}]\label{lem:division}
There exists a constant $c$ such that for every positive integer $r$, 
every planar graph $G$ on $n$ vertices admits an $r$-division $\Rr$ such that 
\begin{equation*}|\Rr|\leq Cn/r\qquad\textrm{and}\qquad\sum_{R\in \Rr} |\prt R| \leq Cn/\sqrt{r}.\end{equation*}
Moreover, given $G$ and $r$ such an $r$-division can be computed in polynomial time.
\end{lemma}

\newcommand{\faccmp}{\fac^\circ}

\subparagraph*{PTAS for {\sc{$k$-Median}} of~\cite{local-search}.} 
We now review the approximation scheme for {\sc{$k$-Median}} of Cohen-Addad et al.~\cite{local-search}, as we will use it as a black-box. 
Formally, we shall need the following statement.

\begin{theorem}[\cite{local-search}]\label{thm:kmedian-ptas}
Given a {\sc{$k$-Median}} instance $(G,\clients,\fac,k)$ where $G$ is planar, a subset of facilities $\faccmp\subseteq \fac$ with $|\faccmp|\leq k$, and an accuracy parameter $\eps>0$, 
it is possible in time $|\fac\setminus \faccmp|^{\Oh(1/\eps^2)}\cdot n^{\Oh(1)}$ to compute a solution $\openfac$ with $\faccmp\subseteq \openfac$
whose connection cost is at most $(1+\eps)$ times larger than the minimum possible connection cost of a solution that contains~$\faccmp$.
\end{theorem}

The statement of Theorem~\ref{thm:kmedian-ptas} somewhat differs from the one presented in~\cite{local-search}; let us review the differences.

First, the result of~\cite{local-search} works in a larger generality, when the graph $G$ is drawn from any fixed proper minor-closed class; we do not need this generality and we focus on the class of planar graphs.

Second, for the running time, the algorithm proposed by Cohen-Addad et al.~\cite{local-search} is actually a simple local search of radius $\Oh(1/\eps^2)$ that stops whenever it cannot find an improvement step
that would decrease the cost by a multiplicative factor of at least $(1+\eps)$. Observe that since in an improvement step we can add or remove only facilities from $\fac\setminus \fac_0$,
within local search radius $\Oh(1/\eps^2)$ there are $|\fac\setminus \faccmp|^{\Oh(1/\eps^2)}$ possible improvement steps, and evaluating each of them takes polynomial time.
Finally, as argued in~\cite{local-search}, the algorithm terminates within $\Oh(|\clients|/\eps)$ steps, so the claimed running time follows.

Third, in the statement of Theorem~\ref{thm:kmedian-ptas} we assume that there is a set $\faccmp$ of {\em{compulsory}} facilities that are required to be taken to the solution.
While this is not stated in~\cite{local-search}, it is straightforward to add this feature to the result.
In the algorithm we start with $\faccmp$ as an original solution and we consider only local search steps that keep it intact.
It is straightforward to see that the analysis of the approximation ratio still holds.
In principle, the analysis relies on showing that if the current solution $\openfac$ is more expensive by at least a multiplicative
factor of $(1+\eps)$ than the optimum solution $\openfac_0$, 
then there is a mixed solution $\openfac'$ that is cheaper than $\openfac$ and the symmetric difference of $\openfac$ and $\openfac'$ has size $\Oh(1/\eps^2)$.
It then suffices to observe that if $\openfac$ and $\openfac_0$ both contain $\faccmp$, then so does the mixed solution~$\openfac'$.

\section{Facility coreset for \textsc{$k$-Median} in planar graphs}\label{sec:main}

In this section we give a coreset for centers for the {\sc{$k$-Median}} problem, that is, we prove Theorem~\ref{thm:fac-coreset}.
We shall focus on the following lemma, which in combination with Theorem~\ref{thm:coreset} yields Theorem~\ref{thm:fac-coreset}.

\begin{lemma}\label{lem:kmedian}
Given a \textsc{$k$-Median} instance $(G,\fac,\clients,k)$ with a weight function $\wei$
and an accuracy parameter $\varepsilon > 0$, 
      one can in polynomial time compute a set $\fac_0 \subseteq \fac$ of size
      $|\supp(\wei)| \cdot (\varepsilon^{-1} \log |V(G)|)^{\Oh(\varepsilon^{-3})}$ satisfying the following condition with constant probability: there exists a set
      $\openfac_0 \subseteq \fac_0$ of size at most $k$ such that for every 
      set $\openfac \subseteq \fac$ of size at most $k$ it holds that
      $\conncost(\openfac_0, \wei) \leq (1+\varepsilon)\cdot \conncost(\openfac, \wei)$.
\end{lemma}

Before we proceed, let us verify that Theorem~\ref{thm:coreset} and Lemma~\ref{lem:kmedian} together imply Theorem~\ref{thm:fac-coreset}.
Given an instance $(G,\fac,\clients,k)$ of {\sc{$k$-Median}}, we first apply Theorem~\ref{thm:coreset} to obtain a coreset $\wei$ with support of size $\Oh(k\eps^{-2} \log n)$.
Next, we pass this coreset to Lemma~\ref{lem:kmedian}, thus obtaining a set $\fac_0\subseteq \fac$ of size $k\cdot (\eps^{-1}\log n)^{\Oh(\eps^{-3})}$.
Let $\openfac_0$ be the subset of $\fac_0$ of size at most $k$ that minimizes $\conncost(\openfac_0,\wei)$.
Then using the approximation guarantees of Theorem~\ref{thm:coreset} and Lemma~\ref{lem:kmedian}, for any $\openfac\subseteq \fac$ we have
\begin{equation*}\conncost(\openfac_0,\clients)\leq (1+\eps)\conncost(\openfac_0,\wei)\leq (1+\eps)^2 \conncost(\openfac,\wei)\leq (1+\eps)^3 \conncost(\openfac,\clients).\end{equation*}
It remains to rescale $\eps$. Hence, for the rest of this section we focus on proving Lemma~\ref{lem:kmedian}.

\medskip

Let $\inst = (G,\fac,\clients,k)$ be an input \textsc{$k$-Median} instance
with a weight function $\wei$, where $G$ is planar. 
Let $\varepsilon > 0$ be an accuracy parameter and without loss of generality assume that $\varepsilon < 1/4$.
Let $n = |V(G)|$ and $m = |E(G)|$. Without loss of generality assume that $n = \Theta(m)$.

We assume that $G$ is embedded in a sphere $\Sigma$ and apply the necessary modifications explained in the beginning of Section~\ref{sec:voronoi} to fit into the framework of Voronoi diagrams.
Denote $S = \supp(\wei)$. We compute the Voronoi partition $\VorPrt_S$ induced by $S$ and the Voronoi diagram $\VorDiag_S$ induced by $S$.
By Proposition~\ref{prop:diagram}, $\VorDiag_S$ has $\Oh(|S|)$ vertices, faces, and edges.

\subparagraph*{Distance levels.}
We first compute an $\Oh(1)$-approximate solution $\apxfac \subseteq \fac$ using the algorithm given by Feldman and Langberg~\cite[Theorem 15.1]{coreset}; 
this algorithm outputs an $\Oh(1)$-approximate solution with constant probability.
Let us scale all the edge lengths in $G$ by the same ratio so that
\begin{equation}\label{eq:conncost-apx}
\conncost(\apxfac, \wei) = |S|/\varepsilon.
\end{equation}
Next, we assign length $+\infty$ to every edge of length larger than $\conncost(\apxfac, \wei)$;
clearly, they are not used in the computation of the connection cost of an optimum solution.
Without loss of generality we assume that all the distances between vertices in $G$ are finite:
otherwise we can split the instance into a number of independent ones, compute a suitable set $\fac_0$ for each of them and take the union.

The next step is to assign levels to distances in the graph.
For any $c \in [0, +\infty)$, define the \emph{level} of $c$, denoted $\lvl(c)$, to be the smallest nonnegative integer $\ell$ such that $c<(1+\eps)^\ell$.
Note that $\lvl(c)=0$ if and only if $c<1$.
Let $\maxlvl=1+\lvl(m\cdot \conncost(\apxfac, \wei))$, then we have
\begin{equation*}
\lvl(\dist(u,v))\in \{0,1,\ldots,\maxlvl-1\}\qquad \textrm{for all }u,v\in V(G).
\end{equation*}
Observe that since $m=\Theta(n)$, by~\eqref{eq:conncost-apx} we have 
\begin{equation}\label{eq:maxlvl}
\maxlvl\leq \Oh(\eps^{-1}\log (m|S|/\eps))\leq \Oh(\eps^{-2}\log n).
\end{equation}


\newcommand{\frth}{\lambda}

\subparagraph*{Portals and profiles.}
Let $\tau = (p, u, f)$ be an incidence in $\VorDiag_S$. 
Let $d(\tau) = \dist(p, u)$ and let $\ell(\tau)=\lvl(d(\tau))$; note that $\Spoke(\tau)$ has length exactly $d(\tau)$.
For every integer $\iota\in \{1,\ldots,\ell(\tau)\}$, 
we define the \emph{portal} $\port{\tau, \iota}$ as a vertex on $\Spoke(\tau)$ at distance exactly $(1+\varepsilon)^{\iota-1}$ from $p$;
we subdivide an edge an create a new vertex to accommodate $\port{\tau, \iota}$ if necessary. 
Furthermore, we add also a portal $\port{\tau, 0} = p$.
Since $\ell(\tau)=\lvl(d(\tau))<\maxlvl$, there are at most $\maxlvl$ portals on the spoke $\Spoke(\tau)$.

Consider a diamond $\Diam(e)$ induced by some edge $e$ of $\VorDiag_S$, and a vertex $v$ in $\Diam(e)$.
Recall that the perimeter of $\Diam(e)$ consists of spokes $\Spoke(\tau_{i,j})$ for four incidence $\tau_{i,j}$, where $i,j \in \{1,2\}$. 
The \emph{profile} of a vertex $w$ belonging to the diamond $\Diam(e)$ consists of the following information, for all $\tau\in \{\tau_{i,j}\colon i,j \in \{1,2\}\}$:
\begin{enumerate}
\item The minimum index $\frth\in \{0,1,\ldots,\ell(\tau)\}$ satisfying 
      \begin{equation*}\dist(\port{\tau,\frth},p) > \varepsilon\cdot \dist(\port{\tau,\frth}, w),\end{equation*}
      where $p=\port{\tau,0}$ is the vertex of $S$ involved in $\tau$.
      If no such index exists, we set $\frth=\ell(\tau)$.
\item Letting
      \begin{equation*}I=\left(\{0\}\cup \{\iota \colon |\iota-\frth|\leq 1/\eps^3\}\right)\cap \{0,1,\ldots,\ell(\tau)\},\end{equation*}
the profile records the value of $\lvl(\dist(w,\port{\tau,\iota}))$ for all $\iota\in I$. 
\end{enumerate}
Whenever speaking about a vertex $w$ and incidence $\tau$, we use $\frth(\tau,w)$ and $I(\tau,w)$ to denote $\frth$ and $I$ as above.
We note that in total there are only few possible profiles.

\begin{claim}\label{cl:profiles-number}
The number of possible different profiles of vertices in $\Diam(e)$ is $\maxlvl^{\Oh(\varepsilon^{-3})}$. 
\end{claim}
\begin{clproof}
Since $0\leq \ell(\tau)<\maxlvl$ for every incidence $\tau$, there are at most $\maxlvl^4$ choices for the four values $\frth(\tau_{i,j},w)$ for $i,j\in \{1,2\}$.
Further, we have $|I(\tau_{i,j},w)|\leq \Oh(\varepsilon^{-3})$, so there are at most $\maxlvl^{\Oh(\varepsilon^{-3})}$ 
choices for the values $\lvl(\dist(w,\port{\tau_{i,j},\iota}))$ for $i,j\in \{1,2\}$ and $\iota\in I(\tau_{i,j},w)$.
\end{clproof}

For future reference, we state the key property of profiles: having the same profile implies having approximately same distances to the profiles with indices in $I$.

\begin{claim}\label{cl:profiles-apx}
Suppose $w$ and $w'$ are two vertices of $\Diam(e)$ that have the same profile.
Then for each $\tau\in \{\tau_{i,j}\colon i,j\in \{1,2\}\}$ and $\iota\in I(\tau,w)$, we have
\begin{equation*}\dist(w',\port{\tau,\iota})\leq (1+\eps)\cdot \dist(w,\port{\tau,\iota})+1.\end{equation*}
\end{claim}
\begin{clproof}
Let $\ell=\lvl(\dist(w,\port{\tau,\iota}))=\lvl(\dist(w',\port{\tau,\iota}))$, as recorded in the common profile.
If $\ell=0$, then $\dist(w',\port{\tau,\iota})<1$ and we are done.
Otherwise, $\dist(w,\port{\tau,\iota})$ and $\dist(w',\port{\tau,\iota})$ are both contained in the interval $[(1+\eps)^{\ell-1},(1+\eps)^{\ell})$.
This interval has length $\eps\cdot (1+\eps)^{\ell-1}\leq \eps \cdot \dist(w,\port{\tau,\iota})$, hence the claim follows.
\end{clproof}

\subparagraph*{Construction of the set $\fac_0$.}
We now construct the set $\fac_0$ as follows: for every diamond $\Diam(e)$ and every
possible profile in $\Diam(e)$, include in $\fac_0$ one facility with that profile (if one exists).
Since there are $\Oh(|S|)$ diamonds, by Claim~\ref{cl:profiles-number} and~\eqref{cl:profiles-number} we have
\begin{equation*}|\fac_0|\leq \Oh(|S|)\cdot \maxlvl^{\Oh(\eps^{-3})}=|\supp(\wei)|\cdot (\varepsilon^{-1} \log n)^{\Oh(\eps^{-3})},\end{equation*}
as claimed. It remains to prove that $\fac_0$ has the claimed approximation properties.

For every facility $w \in \fac$, pick a diamond $\Diam(e)$ containing $w$ and let $f(w)$
to be the facility $f(w) \in \fac_0 \cap \Diam(e)$ that has the same profile as $w$.
Fix a solution $\optfac \subseteq \fac$ with $|\optfac|\leq k$ minimizing $\conncost(\optfac, \wei)$.
Let $\openfac_0 = \{f(w)\colon w \in \optfac\}$. Clearly, $|\openfac_0| \leq |\optfac| \leq k$.
To finish the proof of Lemma~\ref{lem:kmedian} it suffices to show that 
\begin{equation}\label{eq:fpas:1}
\conncost(\openfac_0, \wei) \leq (1+ \Oh(\varepsilon))\conncost(\optfac, \wei).
\end{equation}

To this end, consider any client $v \in S = \supp(\wei)$ 
and let $w \in \optfac$ be the facility in $\optfac$ serving $v$, that is, $\dist(v, w) = \dist(v, \optfac)$.
To show~\eqref{eq:fpas:1}, it suffices to prove that
\begin{equation}\label{eq:fpas:2}
\dist(v, f(w)) \leq (1+\Oh(\varepsilon)) \dist(v, w) + \Oh(1).
\end{equation}
Indeed, by summing \eqref{eq:fpas:2} through all $v\in S$ and using~\eqref{eq:conncost-apx} we obtain
\begin{eqnarray*}
\conncost(\openfac_0, \wei) & \leq & (1+\Oh(\varepsilon)) \conncost(\optfac, \wei) + \Oh(1)\cdot |S|\\
                            & \leq & (1+\Oh(\varepsilon)) \conncost(\optfac, \wei) + \Oh(\eps)\cdot \conncost(\apxfac, \wei)\leq (1+\Oh(\varepsilon)) \conncost(\optfac, \wei),
\end{eqnarray*}
where the last inequality is due to $\apxfac$ being an $\Oh(1)$-approximate solution.

Hence, from now on we focus on proving~\eqref{eq:fpas:2}.
Let $\Diam(e)$ be the diamond containing $w$ and $f(w)$.
Consider the shortest path $P$ from $w$ to $v$ in $G$.
By Lemma~\ref{lem:perimeter-intersects}, the path $P$ intersects the perimeter of the diamond $\Diam(e)$.
Let $u$ be the vertex on the perimeter of $\Diam(e)$ that lies on $P$ and, among such, is closest to $w$ on $P$. 
Since $P$ is a shortest path, the length of the subpaths of $P$ between $v$ and $u$ and between $u$ and $w$ equal $\dist(v,u)$ and $\dist(u, w)$, respectively,
and in particular $\dist(v,w)=\dist(v,u)+\dist(u,w)$.

We now observe that to prove~\eqref{eq:fpas:2}, it suffices show the following.
\begin{equation}\label{eq:fpas:3}
\dist(u, f(w)) \leq \dist(u, w) + \Oh(\varepsilon) \dist(v, w) + \Oh(1).
\end{equation}
Indeed, assuming~\eqref{eq:fpas:3} we have
\begin{eqnarray*}
\dist(v, f(w)) & \leq & \dist(v, u) + \dist(u, f(w))\\
               & \leq & \dist(v, u) + \dist(u, w) + \Oh(\varepsilon) \dist(v, w) + \Oh(1)\\
               & =    & \dist(v, w) + \Oh(\varepsilon) \dist(v, w) + \Oh(1).
\end{eqnarray*}
Hence, from now on we focus on proving~\eqref{eq:fpas:3}.

Let $\tau_{i,j}$ for $i,j\in \{1,2\}$, be the four incidences involved in the diamond $\Diam(e)$.
Since $u$ lies on the perimeter of $\Diam(e)$, actually $u$ lies on $\Spoke(\tau)$, where $\tau=\tau_{i,j}$ for some $i,j\in \{1,2\}$.
Let $p = \port{\tau, 0}$ be the vertex of $S$ involved in the incidence $\tau$.
Since $u \in \VorPrt_S(p)$ while $v \in S$, we have that $\dist(u, p) \leq \dist(u, v)$. 
Consequently, we have $\dist(u,w)\leq \dist(v,w)$ and $\dist(u,p)\leq \dist(v, w)$, so to prove~\eqref{eq:fpas:3} it suffices to prove the following:
\begin{equation}\label{eq:fpas:4}
\dist(u, f(w)) \leq \dist(u, w) + \Oh(\eps)(\dist(u,w) + \dist(u, p)) + \Oh(1).
\end{equation}

Let $\portal_u=\port{\tau, \iota}$ be the portal on the subpath of $\Spoke(\tau)$ between $u$ and $p$ that is closest to $u$.
Intuitively, $\portal_u$ is a good approximation of $u$ and distances from $\portal_u$ are almost the same as distances from $u$.
As this idea will be repeatedly used in this sequel, we encapsulate it in a single claim.

\begin{claim}\label{cl:apx-portal}
Suppose for some vertices $x$ and $y$ we have
\begin{equation*}\dist(\portal_u,x)\leq A\cdot \dist(\portal_u,y)+B,\end{equation*}
for some $A,B$. Then
\begin{equation*}\dist(u,x)\leq A\cdot \dist(u,y)+B+(A+1)\eps\cdot \dist(u,p)+(A+1).\end{equation*}
\end{claim}
\begin{clproof}
By the choice of $\portal_u$ we have 
\begin{equation*}\dist(\portal_u,p)\leq \dist(u,p)\leq (1+\eps)\dist(\portal_u,p)+1,\end{equation*}
so $\dist(u,\portal_u)\leq \eps\cdot \dist(u,p)+1$. Therefore, we have
\begin{eqnarray*}
\dist(u,x) & \leq & \dist(\portal_u,x)+\dist(\portal_u,u) \leq A\cdot \dist(\portal_u,y) + B + \dist(\portal_u,u)\\
           & \leq & A\cdot (\dist(u,y)+\dist(\portal_u,u)) + B + \dist(\portal_u,u)\\
           & =    & A\cdot \dist(u,y) + B + (A+1)\cdot \dist(\portal_u,u)\\
           & \leq & A\cdot \dist(u,y) + B + (A+1)\cdot \eps\cdot \dist(u,p) + (A+1),
\end{eqnarray*}
as claimed.
\end{clproof}


Since $w$ and $f(w)$ have the same profile, we may denote $\frth=\frth(\tau,w)=\frth(\tau,f(w))$ and $I=I(\tau,w)=I(\tau,f(w))$.
Further, let $\portal_\frth=\port{\tau,\frth}$
We now consider a number of cases depending on the relative values of $\iota$ and $\frth$, 
with the goal on proving that~\eqref{eq:fpas:4} holds in each case. See Figure~\ref{fig:vupw} for an illustration.

\begin{figure}[h]
\begin{center}
\includegraphics{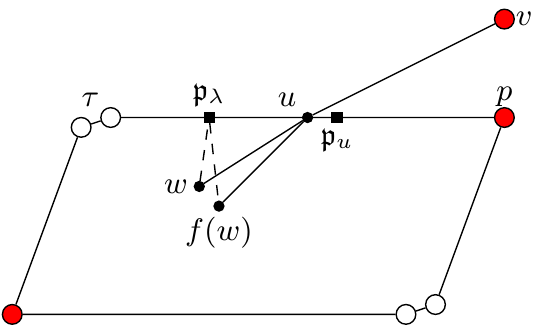}
\end{center}
\caption{The diamond $\Diam(e)$ with vertices $v$, $u$, $w$, $p$, and $f(w)$.
  Red vertices are clients, black squares are portals.
  The case distinction in the proof corresponds to relative order of $\mathfrak{p}_\lambda$ and $\mathfrak{p}_\iota$.}\label{fig:vupw}
\end{figure}

\subparagraph*{Middle case: $\iota\in I$.}
As profiles of $w$ and $f(w)$ are the same and $\iota\in I$, by Claim~\ref{cl:profiles-apx} we have
\begin{equation*}\dist(\portal_u, f(w))\leq (1+\varepsilon)\dist(\portal_u, w)+1.\end{equation*}
It suffices now to apply Claim~\ref{cl:apx-portal} to infer inequality~\eqref{eq:fpas:4}.

\subparagraph*{Close case: $\iota < \frth - 1/\varepsilon^3$.}
Since $\iota < \frth - 1/\varepsilon^3$, by the choice of $\frth$ we have
\begin{equation*}\dist(\portal_u, p) \leq \varepsilon\cdot \dist(\portal_u, w).\end{equation*}
By applying Claim~\ref{cl:apx-portal}, we infer that
\begin{equation*}\dist(u, p) \leq \eps\cdot\dist(u, w) + (1+\eps)\eps\cdot \dist(u,p) + \Oh(1)\leq \eps\cdot \dist(u, w) + \dist(u,p)/2 + \Oh(1),\end{equation*}
which entails
\begin{equation*}\dist(u, p) \leq 2\eps\cdot \dist(u, w)+\Oh(1).\end{equation*}
Since $0\in I$ and the profiles of $w$ and $f(w)$ are equal, by Claim~\ref{cl:profiles-apx} we have
\begin{equation*}\dist(p, f(w)) \leq (1+\varepsilon)\cdot \dist(p, w)+1.\end{equation*}
By combining the last two inequalities, we conclude that
\begin{eqnarray*}
\dist(u, f(w)) & \leq & \dist(u, p) + \dist(p, f(w)) \\
    & \leq & \dist(u, p) + (1 + \eps) \dist(p, w) + 1 \\
    & \leq & \dist(u, p) + (1 + \eps) (\dist(u,p)+\dist(u, w)) +1 \\
    & =    & (2 + \eps) \dist(u, p) + (1+\eps) \dist(u, w) + 1\\
    & \leq & (1 + \Oh(\eps)) \dist(u, w) + \Oh(1).
\end{eqnarray*}
Thus, inequality~\eqref{eq:fpas:4} holds in this case.

\subparagraph*{Far case: $\iota > \frth + 1/\varepsilon^3$.}
By the definition of $\frth$ and since $\iota > \frth + 1/\varepsilon^3$, we have in particular $\frth < \ell(\tau)$ and hence
\begin{equation*}\eps\cdot \dist(\portal_\frth, w) < \dist(\portal_\frth, p).\end{equation*}
On the other hand, we have
\begin{equation*}\dist(p, \portal_\frth) = (1+\eps)^{\frth-\iota} \dist(p, \portal_u)\leq (1+\eps)^{-1/\eps^3}\dist(p, \portal_u)\leq \eps^2\cdot \dist(p, \portal_u),\end{equation*}
where the last step follows from Bernoulli's inequality. By combining the above two inequalities we obtain
\begin{equation*}\dist(\portal_\frth, w) \leq \eps\cdot \dist(p, \portal_u),\end{equation*}
implying
\begin{equation*}\dist(p,w)\leq \dist(p,\portal_\frth)+\dist(\portal_\frth, w)\leq (\eps+\eps^2)\cdot \dist(p, \portal_u)\leq 2\eps\cdot \dist(p, \portal_u)\leq 2\eps\cdot \dist(p,u).\end{equation*}
As before, by Claim~\ref{cl:profiles-apx} we have $\dist(p,f(w))\leq (1+\eps)\dist(p,w)+1$ due to $0\in I$, hence
\begin{eqnarray*}
\dist(u, f(w)) & \leq & \dist(u, p) + \dist(p, f(w)) \\
    & \leq & (\dist(u, w) + \dist(p, w)) + (1 + \eps) \dist(p, w) + 1 \\
    & \leq &  \dist(u, w) + (2+\eps)\cdot (2\eps)\cdot \dist(u,p) + 1 \\
    & \leq & \dist(u, w) + \Oh(\varepsilon) \dist(u, p) + \Oh(1).
\end{eqnarray*}
We conclude that inequality~\eqref{eq:fpas:4} holds in this case.

\medskip

As the case investigation presented above covers all the possibilities, the proof of Lemma~\ref{lem:kmedian} is complete, so we have also proved Theorem~\ref{thm:fac-coreset}.



\section{Bicriteria EPTAS for \textsc{$k$-Median} in planar graphs}
\label{sec:kmedian-bi}

With all the tools prepared, we proceed to the proof of Theorem~\ref{thm:kmedian-bi}.
Let $(G,\clients,\fac,k)$ be an input \textsc{$k$-Median} instance. 
As before, we modify $\inst$ by triangulating each face with infinity-cost edges and slightly perturbing the edge lengths so that the shortest paths
are unique and finite distances are pairwise different. 
As we can rescale $\eps$ at the end, we aim at a solution consisting of $(1+\Oh(\eps))k$ facilities and 
having connection cost at most $(1+\Oh(\eps))$ times larger than the minimum possible connection cost of a solution of size $k$.

\subparagraph*{Step 1: Compute a facility coreset and contract the graph.}
The first step of the algorithm is to apply Theorem~\ref{thm:fac-coreset} to $(G,\clients,\fac,k)$ and $\eps$, 
thus obtaining a subset of facilities $\fac_0\subseteq \fac$ of size $k\cdot (\log n/\eps)^{\Oh(1/\eps^3)}$
that contains a $(1+\eps)$-approximate solution.

We next compute the Voronoi partition $\{\VorPrt_{\fac_0}(p)\}_{p\in \fac_0}$ induced by $\fac_0$ in $G$.
Define the {\em{contracted graph}} $H$ as the graph on the vertex set $\fac_0$ where two vertices $p,q\in \fac_0$ are adjacent if and only if in $G$ there is an edge with one endpoint in 
$\VorPrt_{\fac_0}(p)$ and second in $\VorPrt_{\fac_0}(q)$.
Note that $H$ is unweighted and connected.
Since Voronoi cells induce connected subgraphs of $G$, the graph $H$ can be obtained from $G$ by contracting the whole cell $\VorPrt_{\fac_0}(p)$ onto $p$, for each $p\in \fac_0$. This implies that $H$ is planar.

\subparagraph*{Step 2: Compute an $r$-division and solve the regions.}
Set
\begin{equation*}r = (\eps^{-1} |\fac_0|/k)^2 \in (\eps^{-1}\log n)^{\Oh(\varepsilon^{-3})}\end{equation*}
and apply Lemma~\ref{lem:division} to compute in polynomial time an $r$-division $\Rr$ of $H$ satisfying
\begin{equation}\label{eq:boundaries}
\sum_{R\in \Rr} |\prt R|\leq \Oh(\eps)\cdot k.
\end{equation}
We now partition the client set $\clients$ into sets $(\clients_R)_{R\in \Rr}$ as follows:
take every client $v\in \clients$ and letting $p\in \fac_0$ be such that $v\in \VorPrt_{\fac_0}(p)$
assign $v$ to the set $\clients_R$ for any region $R\in \Rr$ such that $p\in V(R)$. 
Note that for clients residing in cells that are not on the boundary of any region there is exactly one choice for such a region $R$, whereas
clients from boundary cells have more than one option: we choose an arbitrary one so that every client is assigned to exactly one set $\clients_R$.

Now for each $R\in \Rr$ and $\ell\in \{0,1,\ldots,k\}$ construct an instance $\Ii(R,\ell)=(G,\clients_R,V(R),\ell+|\prt R|)$ of {\sc{$k$-Median}}.
For each such instance, apply the algorithm of Theorem~\ref{thm:kmedian-ptas} with the set of compulsory facilities $\prt R$.
This yields a solution $\openfac(R,\ell)\subseteq V(R)$ that contains $\prt R$ and at most $\ell$ other facilities of $V(R)\setminus \prt R$, 
which has connection cost at most $(1+\eps)$ times the optimum connection cost of a solution satisfying these conditions.
Let $\cost(R,\ell)=\conncost(\openfac(R,\ell),\clients_R)$.
Observe that $|R\setminus \prt R|\leq r\leq (\eps^{-1}\log n)^{\Oh(\varepsilon^{-3})}$, hence by Theorem~\ref{thm:kmedian-ptas} the running time needed for solving each individual instance $\Ii(R,\ell)$ is
\begin{equation*}(\eps^{-1}\log n)^{\Oh(\varepsilon^{-3})\cdot \Oh(\eps^{-2})}\cdot n^{\Oh(1)}\leq 2^{\Oh(\eps^{-5}\log(\varepsilon^{-1}))}\cdot n^{\Oh(1)},\end{equation*}
where again we use the fact that $(\log n)^d\leq 2^{\Oh(d\log d)}\cdot n^{\Oh(1)}$. 
There are at most $n^2$ instances $\Ii(R,\ell)$ to solve, hence the total running time spent is again $2^{\Oh(\eps^{-5}\log(\varepsilon^{-1}))}\cdot n^{\Oh(1)}$.

\subparagraph*{Step 3: Assemble the regional solutions.}
Finally, consider the following problem: find a vector $(\ell_R)_{R\in \Rr}$ with $\sum_{R\in \Rr} \ell_R\leq k$ that minimizes $\sum_{R\in \Rr} \cost(R,\ell_R)$.
It is straightforward to see that this problem can be solved in polynomial time by a standard knapsack dynamic programming as follows: order the regions arbitrarily and
iterate through the order while keeping a dynamic programming table that for each $k'\in \{0,1,\ldots,k\}$ keeps the minimum possible cost that can be obtained among regions $R$ scanned so far 
with $\ell_R$s summing up to~$k'$.

Having computed an optimum solution $(\ell_R)_{R\in \Rr}$ to the problem above, construct the output solution
\begin{equation*}\openfac = \bigcup_{R\in \Rr} \openfac(R,\ell_R).\end{equation*}
Observe that by~\eqref{eq:boundaries} we have
\begin{equation*}|\openfac|\leq\sum_{R\in \Rr} |\openfac(R,\ell_R)|= \sum_{R\in \Rr} |\prt R| + \sum_{R\in \Rr} \ell_R\leq \Oh(\eps)\cdot k + k = (1+\Oh(\eps))k,\end{equation*}
as required. Along the description we argued that the running time of the algorithm is $2^{\Oh(\eps^{-5}\log(\eps^{-1}))}\cdot n^{\Oh(1)}$.
It remains to bound the connection cost of $\openfac$.

\subparagraph*{Approximation factor.}
Let $\optfac$ be an optimum solution to $(G,\clients,\fac,k)$.
By the properties of $\fac_0$ asserted by Theorem~\ref{thm:fac-coreset}, there exists a solution
$\optfac_0 \subseteq \fac_0$ with $|\optfac_0|\leq k$ satisfying
\begin{equation}\label{eq:core}
\conncost(\optfac_0,\clients)\leq (1+\eps)\cdot \conncost(\optfac,\clients).
\end{equation}
Let $\optfac_1 = \optfac_0\cup \bigcup_{R\in \Rr}\prt R$.
Clearly
\begin{equation}\label{eq:01}
\conncost(\optfac_1,\clients)\leq \conncost(\optfac_0,\clients),
\end{equation}
because $\optfac_1\supseteq \optfac_0$.
The next claim is crucial: due to buying all the facilities in the boundaries of all the regions, in $\optfac_1$ the client-facility connections are isolated into regions.

\begin{claim}\label{cl:kmedian:sep}
For every $R \in \Rr$ and every client $v\in \clients_R$, 
the facility of $\optfac_1$ that is the closest to $v$ belongs to $V(R)$.
\end{claim}
\begin{clproof}
Let $p \in \optfac_0$ be the facility closest to $v$ and, assuming the contrary, suppose $p\notin V(R)$.
Let $P$ be the shortest path from $v$ to $p$ in $G$.
By the construction of $H$ and the properties of division $\Rr$, supposition $p\notin V(R)$ entails that
$P$ traverses a vertex $u$ that belongs to $\VorPrt_{\fac_0}(q)$ for some vertex $q \in \prt R$. 
However, as the Voronoi partition is defined with respect to $\fac_0$
   and $p,q \in \fac_0$, we have that $\dist(u, q) < \dist(u, p)$, implying $\dist(v,q) < \dist(v,p)$.
This is a contradiction to the choice of $p$.
\end{clproof}

For each $R\in \Rr$, let $\optfac_R = \optfac_1\cap V(R)$ and let $\ell^\star_R=|\optfac_R\setminus \prt R|$. Note that $\prt R\subseteq \optfac_R$.
By the approximation guarantee of the algorithm of Theorem~\ref{thm:kmedian-ptas}, we have
\begin{equation}\label{eq:region-apx}
\cost(R,\ell^\star_R)\leq (1+\eps)\cdot\conncost(\optfac_R,\clients_R)\qquad\textrm{for each }R\in \Rr.
\end{equation}
As $\{\clients_R\}_{R\in \Rr}$ is a partition of $\clients$, by Claim~\ref{cl:kmedian:sep} we infer that
\begin{equation}\label{eq:region-sum}
\conncost(\optfac_1,\clients)=\sum_{R\in \Rr}\conncost(\optfac_R,\clients_R).
\end{equation}
Note that sets $\optfac_R\setminus \prt R$ for $R\in \Rr$ are pairwise disjoint and contained in $\optfac_0$, hence $\sum_{R\in \Rr} \ell^\star_R\leq |\optfac_0|\leq k$.
As $(\ell_R)_{R\in \Rr}$ was an optimum solution of the final knapsack problem, we have
\begin{equation}\label{eq:knapsack}
\sum_{R\in \Rr} \cost(R,\ell_R)\leq \sum_{R\in \Rr} \cost(R,\ell^\star_R).
\end{equation}
Finally, the construction of $\openfac$ immediately yields
\begin{equation}\label{eq:openfac}
\conncost(\openfac,\clients)\leq \sum_{R\in \Rr} \cost(R,\ell_R).
\end{equation}
By combining~\eqref{eq:core},~\eqref{eq:01},~\eqref{eq:region-apx},~\eqref{eq:region-sum},~\eqref{eq:knapsack}, and~\eqref{eq:openfac} we conclude that
\begin{equation*}\conncost(\openfac,\clients)\leq (1+\eps)^2\cdot \conncost(\optfac,\clients)\leq (1+\Oh(\eps))\cdot \conncost(\optfac,\clients).\end{equation*}
This finishes the proof.

\bibliography{ref} 

\end{document}